\documentclass[12pt,oneside]{amsart}

\usepackage[english]{babel}
\usepackage{cite}
\usepackage{latexsym}
\usepackage{amssymb,amsthm,amsmath}
\usepackage{hyperref}
\usepackage{graphicx}
\usepackage{longtable}
\usepackage{xcolor}
\usepackage{enumerate}

\newtheorem{theorem}{Theorem}[section]
\newtheorem{proposition}[theorem]{Proposition}
\newtheorem{lemma}[theorem]{Lemma}
\newtheorem{corollary}[theorem]{Corollary}
\newtheorem{remark}[theorem]{Remark}
\newtheorem{problem}{Problem}

{\theoremstyle{definition}
\newtheorem*{definition*}{Definition}
\newtheorem{definition}[theorem]{Definition}}
\newtheorem*{proposition*}{Proposition}
\newtheorem*{corollary*}{Corollary}
\newtheorem*{lemma*}{Lemma}
\newtheorem*{remark*}{Remark}

\newcommand{\F}{\mathbb {F}}

\title{Minimal linear codes arising from blocking sets}
\author{Matteo Bonini}
\thanks{M. Bonini is with Department of Mathematics, University of Trento, Trento, Italy, \url{matteobonini11@gmail.com}}
\author{Martino Borello}
\thanks{M. Borello is with LAGA,  UMR 7539, CNRS, Universit\'e Paris 13 - Sorbonne Paris Cit\'e, Universit\'e Paris 8, F-93526, Saint-Denis, France, \url{borello@math.univ-paris13.fr}}
\date{}

\begin{document}
    \begin{abstract}
        Minimal linear codes are algebraic objects which gained interest in the last twenty years, due to their link with Massey's secret sharing schemes. In this context, Ashikhmin and Barg provided a useful and a quite easy to handle sufficient condition for a linear code to be minimal, which has been applied in the construction of many minimal linear codes.

        In this paper, we generalize some recent constructions of minimal linear codes which are not based on Ashikhmin--Barg's condition. More combinatorial and geometric methods are involved in our proofs. In particular, we present a family of codes arising from particular blocking sets, which are well-studied combinatorial objects. In this context, we will need to define cutting blocking sets and to prove some of their relations with other notions in blocking sets' theory. At the end of the paper, we provide one explicit family of cutting blocking sets and related minimal linear codes, showing that infinitely many of its members do not satisfy the Ashikhmin--Barg's condition.
    \end{abstract}

\maketitle


\section{Introduction}
A codeword of a linear code is called \emph{minimal} if its support
(the set of nonzero coordinates) does not contain the support of any
other linearly independent codeword. For a linear code
$\mathcal{C}$, the supports of minimal codewords of the dual code
$\mathcal{C}^\perp$ give the access structure of a secret sharing
scheme, introduced by Massey in \cite{M93},\cite{M1995}. So, it is
particularly interesting to determine all minimal codewords of a
code. The problem of describing the set of minimal codewords of a
linear code is quite difficult in general, even in the binary case:
actually, the knowledge of the minimal codewords is related with the
complete decoding problem which is known to be NP-hard
\cite{BMeT1978}, even if preprocessing is allowed \cite{BN1990}. To
simplify this task, one can try to find linear codes for which all
codewords are minimal, called \emph{minimal linear codes}. The
problem of finding minimal linear codes has been first investigated
in \cite{DY2003}. Besides their use in secret sharing schemes,
minimal linear codes have shown to be suitable to other
applications: for example, in \cite{CCP2013}, minimal linear codes
are used to ensure privacy in a secure two-party computation. One of
the central results in this context is due to Ashikhmin and Barg
\cite{AB1998}: a linear code $\mathcal{C}$ over a finite field
$\mathbb{F}_q$ of order $q$ is minimal if
\begin{equation}\tag{AB}\label{AB}
    \frac{w_{max}}{w_{min}} < \frac{q}{q-1}
\end{equation}
where $w_{min}$ and $w_{max}$ respectively denote the minimum and
maximum nonzero weights in $\mathcal{C}$. This sufficient condition
(which we will call AB condition) provides an easy criterion to
construct minimal linear codes, especially in the case of codes with
few nonzero weights (see for example \cite{DD},\cite{MOS}). The AB
condition is not necessary. As remarked in \cite{HDZ2018}, a linear
code $\mathcal{C}$ over $\mathbb{F}_q$ is minimal if and only if for
any linearly independent codewords $c,c'\in \mathcal{C}$
\[\sum_{a\in\mathbb{F}_q^*}{\rm wt}(c'-ac)\ne (q-1)\cdot {\rm wt}(c')-{\rm wt}(c).\]
Minimal linear codes not satisfying the Ashkhmin-Barg  condition (AB
condition) are presented in  \cite{CH2017},\cite{DHZ2018},\cite{ZYW}
for the binary case, in \cite{HDZ2018} for the ternary case, and in
\cite{BB2019} for a general odd prime power. One of the main tools
used in these papers is the investigation of the Walsh spectrum of
generalized Boolean functions, which is used to characterize when
linear codes from a general construction are minimal. The aim of the
present paper is to describe in full generality the construction
given in \cite{HDZ2018},\cite{BB2019}, to have more precise
conditions and other families of examples. More combinatorial and
geometric methods are involved in our proofs. In particular, we
present a family of codes arising from particular blocking sets,
which are well-studied combinatorial objects. We believe that the up
to now unexplored link between minimal linear codes and blocking
sets may give rise to many families of new minimal linear codes and
maybe to new perspectives in blocking sets' theory. In the paper we
present, as an example, one particular family of blocking sets and
related minimal linear codes and we prove that  infinitely many
members of this family do
not satisfy the AB condition.\\
The paper is structured as follows: in Section \ref{sec:background}
we recall the main notions necessary to understand the paper; in
Section \ref{sec:bs} we recall some results about blocking sets and
we introduce new definitions and prove new results in this context;
Section \ref{sec:codes} is devoted to the general construction of
minimal linear codes related to blocking sets, both in affine and in
projective case; in Section \ref{sec:functions} we present an
explicit family of minimal linear codes and we show that infinitely
many of its members do not satisfy AB condition; finally in Section
\ref{sec:conclusion} we resume our results and we present some open
problems.
\bigskip
\section{Background}\label{sec:background}
In the whole paper, $\mathbb{F}_q$ will be a finite field with $q$
elements and $e_1,\ldots,e_n$ will denote the canonical basis of
$\F_q^n$.
 \subsection{Linear codes}
    We recall here some basic definition in coding theory (for a complete exposition of the concepts the reader is referred to \cite{HP}).
    \begin{definition}
     Let  $k,n$ be two positive integers such that $k \leq n$. Let $\mathcal{C}$ be a $k$-dimensional vector subspace of $\mathbb{F}_q^n$: we say that $\mathcal{C}$ is a \emph{$q$-ary linear code} of dimension $k$ and length $n$ or an $[n,k]$ code over $\F_q$. The elements of $\mathcal{C}$ are usually called \emph{codewords}. A \emph{generator matrix} of $\mathcal{C}$ is a matrix whose rows form a basis of $\mathcal{C}$ as a vector space over $\F_q$.
    \end{definition}
    Classically, a linear code is endowed with the Hamming metric.
    \begin{definition}
    For $x \in \mathbb{F}_q^n$, the \emph{Hamming weight} of $x$ is the number of nonzero coordinates of $x$ (i.e. the cardinality of its \emph{support} ${\rm supp}(x)$). We denote it by ${\rm wt}(x)$.
    \end{definition}
    For $v=(v_1,\ldots,v_n),w=(w_1,\ldots,w_n)\in \F_q^n$, we denote
    $$v\cdot w=\sum_{i=1}^nv_iw_i$$
    the Euclidean inner product between $v$ and $w$. For a code $\mathcal{C}$, the \emph{dual code} is defined as $\mathcal{C}^\perp:=\{v\in \F_q^n \mid v\cdot c=0, \ \forall c\in \mathcal{C}\}$.

    The following is a central definition for this paper.

    \begin{definition} A codeword $c\in \mathcal{C}$ is \emph{minimal} if it covers only linearly dependent  codewords, i.e. if for all  $c^{\prime}\in \mathcal{C}$,
    \[{\rm supp}(c) \subset {\rm supp}(c^{\prime}) \Longrightarrow \exists \lambda \in \mathbb{F}_q^* \ : \ c^{\prime}=\lambda c.\]
    A linear code $\mathcal{C}$ is \emph{minimal} if every nonzero codeword $c \in \mathcal{C}$ is minimal.
    \end{definition}

    The symmetric group of degree $n$ acts naturally on $\F_q^n$ permuting the coordinates of vectors. This action induce an action on linear codes, which preserves length, dimension and weights. Codes in the same orbits are called \emph{equivalent}. There are more general concepts of equivalence, but they are not necessary for this paper.

\subsection{Affine and projective hypersurfaces}

We recall here some basic definitions of affine and projective
geometry over finite fields. For a more detailed introduction, we
refer to \cite{HirschBook}.

Let $\mathbb{A}(\F_{q}^{n})$ be the affine space of dimension $n$
over the field $\mathbb{F}_q$ and $\mathbb{P}(\F_q^n)$ the
projective space of dimension $n-1$ over $\F_q$.

    Let $f:\F_{q}^n\to\F_{q}$ be a function. It is well-known that all the functions from $\F_{q}^n$ to $\F_{q}$ are polynomial functions (see for example \cite{Clark}). Let us call $$V(f)=\{x\in\F_{q}^n\,|\,f(x)=0\}\subset \mathbb{A}(\F_{q}^{n})$$
    the \emph{affine hypersurface} defined by $f$. We denote  $V(f)^*=V(f)\setminus\{0\}$. If $f$ is homogeneous, $V(f)$ can be seen also as an hypersurface in $\mathbb{P}(\F_q^n)$ by identifying linearly dependent vectors. To avoid any risk of confusion, we will denote $V_p(f)$ the hypersurface in $\mathbb{P}(\F_q^n)$.
    When $f$ is linear, then $V(f)$ is called \emph{affine hyperplane}.
    In this paper, we will mainly consider hyperplanes through the origin, that is hyperplanes for which $f$ is linear and homogeneous. These can be described in terms of a orthogonal vector, as follows (we use the same notation as in \cite{BB2019}).

    \begin{definition}
    Let $v\in\F_q^n$, we define the hyperplane $H(v)$ as the set $$H(v):=\left\{x\in\F_q^n \mid  v\cdot x=0\right\}=\langle v\rangle^\perp,$$
    i.e. $H(v)$ is the set of all the vectors of $\F_q^n$ that are orthogonal to the vector $v$. As done for $V(f)$, we call $H(v)^*:=H(v)\setminus\{0\}$. Note that $H(v)$ can be seen as a hyperplane in $\mathbb{P}(\F_q^n)$ by identifying linearly dependent vectors. Moreover, in the projective case, these are all hyperplanes.
    \end{definition}

\bigskip
\section{Cutting blocking sets}\label{sec:bs}

    We recall some definitions about blocking sets (for a more detailed introduction see \cite[Chapter 3]{SdB}). We will need to introduce some new definitions (Definition \ref{def:vectorial} and Definition \ref{def:cutting}) and prove a new result which adapt well to our context.

    \begin{definition}
    An \emph{affine} (respectively \emph{projective}) $k$-\emph{blocking set} is a subset of an $n$-dimensional affine (respectively projective) space intersecting all $(n-k)$-dimensional affine (respectively projective) subspaces. An affine (respectively projective) $1$-blocking set is also called affine (respectively projective) blocking set.
    \end{definition}

    As we mentioned, it is well-known that every function from $\F_q^n$ to $\F_q$ is polynomial. Every set in the affine space $\mathbb{A}(\F_q^n)$ can be then seen as a hypersurface $V(f)$ for some polynomial function. In particular, every $k$-blocking set can be seen as an hypersurface. With the above definitions, which are the classical ones, there is a substantial asymmetry in the case of $f$ being homogeneous, because if $V(f)^*$ is a $k$-blocking set in the affine space $\mathbb{A}(\F_q^n)$, then $V_p(f)$ is a $k$-blocking set in the projective space $\mathbb{P}(\F_q^n)$, but the vice versa is not true. We think that a natural definition to avoid this asymmetry is the following, which is a weakened version of the affine definition.

    \begin{definition}\label{def:vectorial}
    A \emph{vectorial} $k$-\emph{blocking set} is a subset of an $n$-dimensional affine space not containing the origin intersecting all $(n-k)$-dimensional affine subspaces through the origin. A vectorial $1$-blocking set is also called vectorial blocking set.
    \end{definition}

    A $k$-blocking set $\mathcal{B}$ is \emph{$d$-dimensional} if the subspace generated by $\mathcal{B}$ has dimension $d$. A $k$-dimensional affine (respectively projective, respectively vectorial) subspace is a $k$-blocking set, a $k$-blocking set containing one is called \emph{trivial}.

    \begin{definition}\label{def:bs}
    An affine (respectively projective, respectively vectorial) \emph{$(k, s)$-blocking set}, where $n>k$, is an affine (respectively projective, respectively vectorial) $k$-blocking set that does not contain an affine (respectively projective, respectively affine through the origin) subspace of dimension $s$.
    \end{definition}

    As far as we know, the following property is not equivalent (or at least implied) by any known notion in blocking sets' theory, but it will play a crucial role in our results.

    \begin{definition}\label{def:cutting}
    An affine (respectively projective, respectively vectorial) $k$-blocking set is \emph{cutting} if its intersection with every $(n-k)$-dimensional affine (respectively projective, respectively affine through the origin) subspace is not contained in any other
    $(n-k)$-dimensional affine (respectively projective, respectively affine through the origin) subspace.
    \end{definition}

    The property of being cutting is quite strong, as one can see from the following result.

    \begin{theorem}\label{lemma:ndim}
A set $\mathcal{B}$ is a $t$-dimensional (with $t>n-k$) affine
(respectively projective, respectively vectorial) cutting
$k$-blocking set if and only if the intersection between
$\mathcal{B}$ and every $(n-k)$-dimensional affine (respectively
projective, respectively affine through the origin) subspace is not
contained in any other
    $(n-k)$-dimensional affine (respectively projective, respectively affine through the origin) subspace.
\end{theorem}

\begin{proof}

\begin{itemize}
    \item[$(\Rightarrow)$] by definition.
    \item[$(\Leftarrow)$]  We will prove it only in the affine case to simplify the notations (but in the other cases the proof is exactly the same). If the intersection between $\mathcal{B}$ and every $(n-k)$-dimensional affine subspace is not contained in any other $(n-k)$-dimensional affine subspace, then in particular the intersection is non-empty, so that
$\mathcal{B}$ is an affine $k$-blocking set.

If it exists an $(n-k)$-dimensional affine subspace $\mathcal{S}'$
such that $\mathcal{B}\subset \mathcal{S}'$, then $\mathcal{B}\cap
\mathcal{S}\subset \mathcal{B}\subset \mathcal{S}',$ for any
$(n-k)$-dimensional affine subspace $\mathcal{S}$, a contradiction.
So $\mathcal{B}$ is an $t$-dimensional blocking set, with $t> n-k$.

The cutting property is by definition again.
\end{itemize}
\end{proof}

In Section \ref{sec:codes} we will relate the cutting property to
minimal linear codes and in Section \ref{sec:functions} we will
provide an example of both vectorial and projective cutting blocking
sets.

    \bigskip
    \section{A family of codes arising from cutting blocking sets}\label{sec:codes}

    In this section we take up the notion of the code investigated in \cite{BB2019},\cite{HDZ2018} and prove some new connections with blocking sets.

    For every function $f:\mathbb{F}_q^n\to \mathbb{F}_q$, let $\mathcal{C}_f$ be the linear code defined as
    \begin{equation}\label{Def:Code}
    \mathcal{C}_f :=\{(uf(x)+v\cdot x)_{x \in \mathbb{F}_q^n \setminus \{0\}} \ | \ u \in \mathbb{F}_q, v \in \mathbb{F}_q^n\},
    \end{equation}
    where $v\cdot x$ is the Euclidean inner product between $v$ and $x$, as in \S\ref{sec:background}. Here we are supposing, as usual in coding theory, to have a fixed ordering of $\mathbb{F}_q^n \setminus \{0\}$. What follows will not depend on this ordering: if we change the ordering, we just obtain an equivalent code.

    For any pair $(u,v) \in \mathbb{F}_q\times \mathbb{F}_{q}^n$, let us denote $c(u,v):=(uf(x)+v\cdot x)_{x \in \mathbb{F}_q^n \setminus \{0\}}$. For an $x\in\mathbb{F}_q^n \setminus \{0\}$, we denote by $c(u,v)_x$  the entry in $c(u,v)$ corresponding to $x$. The support ${\rm supp}(c(u,v))$ of a codeword $c(u,v)$ is defined as the set of $\{x \in \mathbb{F}_q^n \setminus \{0\} \ : \ c(u,v)_x\neq 0\}$. The complement of the support $\overline{{\rm supp}}(c(u,v))$ is then the set of $\{x \in \mathbb{F}_q^n \setminus \{0\} \ : \ c(u,v)_x= 0\}$.

\subsection{Parameters of \texorpdfstring{$\mathcal{C}_f$}{Lg}}

\begin{proposition}\label{Prop:parameters}

If $f$ is not linear and $V(f)^*\ne \F_q^n\setminus\{0\}$, the code
$\mathcal{C}_f$ defined as in \eqref{Def:Code}  has length $q^n-1$
and dimension $n+1$ over $\mathbb{F}_{q}$.
    \end{proposition}
    \begin{proof}
        Clearly, the length of $\mathcal{C}_f$ is $\#(\mathbb{F}_q^n\setminus \{0\})=q^n-1$.

        Each codeword in $\mathcal{C}_f$ can be written as linear combination of $c(1,0)$, $c(0,e_1)$, \ldots, $c(0,e_n)$, where $e_1,\ldots,e_n$ is the standard basis of $\mathbb{F}_q^n$ over $\mathbb{F}_q$. Showing that these vectors are linearly independent is equivalent to prove that $c(u,v)=0$ if and only if $u=0$ and $v=0$.
        \begin{itemize}
            \item If $u=0$, then $c(u,v)_{e_i}=v_i=0$ for $i\in \{1,\ldots,n\}$, so that $v=0$.
            \item If $u\neq 0$, $uf(x)+v\cdot x=0$ for all $x\in \F_q^n\setminus \{0\}$ implies that $V(f)=H(v)$, if $v\ne 0$, which is a contradiction.
        \end{itemize}
        This proves that $c(1,0)$, $c(0,e_1)$, \ldots, $c(0,e_n)$ is a basis of $\mathcal{C}$ of size $n+1$.
    \end{proof}

    If $f$ is not linear and $V(f)^*\ne \F_q^n\setminus\{0\}$, we have then that a generator matrix of $\mathcal{C}_f$ is

    \[
    G=\left[\begin{matrix}
    f(x_0)&f(x_1)&\ldots&f(x_{q^n-1})\\
    (x_0)^T&(x_1)^T&\ldots&(x_{q^n-1})^T
    \end{matrix}\right],
    \]
    where $(x_i)^T$, for $i\in\{0,\dots,q^n-1\}$, are the column vectors of $\F_{q}^n\setminus\{0\}$ with respect to the fixed ordering.

    \begin{remark}\label{rmk:weights}
        The matrix $G$ is formed by two blocks, one of them is the line obtained from the evaluations of the function $f$, the other one is the matrix formed by all the ordered vectors of $\F_{q}^n\setminus\{0\}$. The second block is the generator matrix of the simplex code for $q=2$ and, up to a reordering of the coordinates, the concatenation of $q-1$ copies of the simplex code for $q>2$ $($see {\rm \cite[\S 1.8]{HP}}$)$. Therefore, ${\rm wt}(c(0,v))=q^{n}-q^{n-1}$ $($see {\rm \cite[Theorem 2.7.5]{HP}}$)$ and  ${\rm wt}(c(u,0))=q^n-1-\#V(f)^*$.
    \end{remark}

        It seems to be quite difficult to prove general results on the weight distribution of $\mathcal{C}_f$. On the other hand, we can easily observe a condition on $\# V(f)^*$ for $\mathcal{C}_f$ to not satisfy the AB condition.

        \begin{lemma}\label{lemma:AB}
        If $\# V(f)^*\geq 2q^{n-1}-q^{n-2}-1$, then
        \[\frac{w_{max}}{w_{min}} \geq \frac{q}{q-1}.\]     \end{lemma}

        \begin{proof}
        It is enough to use the trivial fact that if    $$w_{\min}\leq w_1\leq w_2\leq w_{\max}$$
        then $$\frac{w_{max}}{w_{min}}\geq \frac{w_2}{w_1},$$
        with the known weights given in Remark \ref{rmk:weights}.
        \end{proof}

    \subsection{Minimality of \texorpdfstring{$\mathcal{C}_f$}{Lg}}

Let $f:\F_q^n\to \F_q$ a function. For reader's convenience, we
recall that, by Theorem \ref{lemma:ndim}, $V(f)^*$ is an
$n$-dimension cutting vectorial blocking set if and only if
 the intersection between $V(f)^*$ and any hyperplane through the origin $H(v)$ is not contained in any other hyperplane through the origin $H(v')$.

Moreover, the blocking set is a $(1,n-1)$-blocking set if and only
if $V(f)^*$ does not contain any hyperplane without the origin
$H(v)^*$.

\begin{remark}\label{rmk:nonlinear}
If $V(f)^*$ does not contain any hyperplane without the origin
$H(v)^*$, then clearly  $f$ is not linear and $V(f)^*\ne
\F_q^n\setminus\{0\}$.
\end{remark}

\begin{lemma}\label{lemma:blocking}
 The intersection between $V(f)^*$ and any hyperplane through the origin $H(v)$ is not contained in any other hyperplane through the origin $H(v')$ if and only if for any $v,v'\in \F_q^n\setminus\{0\}$, with $H(v)\ne H(v')$, it exists $x\in \F_q^n\setminus \{0\}$ such that     \begin{equation}\label{eq:blocking}
     uf(x)+v\cdot x=0 \, \, \text{ and } \, \, u'f(x)+v'\cdot x\neq 0,
 \end{equation}
    for every $u,u'\in \F_q$.
\end{lemma}

 \begin{proof}
   If the intersection between $V(f)^*$ and any hyperplane is not contained in any other hyperplane, then  for any $v,v'\in \F_q^n\setminus\{0\}$, with $H(v)\ne H(v')$, it exists $x\in V(f)^*\cap H(v)$ but not in $H(v')$, which satisfies \eqref{eq:blocking}.
   Conversely, if $x\not \in H(v)$, then $uf(x)+v\cdot x$ cannot be equal to $0$ for all $u\in \F_q$. So $x\in H(v)$, which implies $f(x)=0$, so that $x\in V(f)^*\cap H(v)$, and, finally, $x\not \in H(v')$.
 \end{proof}

\begin{theorem}\label{Th:Main}:
Let $f:\F_q^n\to \F_q$ be a function. If

\begin{enumerate}
    \item[{\rm(a)}] $V(f)^*$ is  an $n$-dimensional cutting vectorial $(1,n-1)$-blocking set  in
$\mathbb{A}(\F_q^n)$;
\item[{\rm (b)}] for every non-zero vector $v$, it exists  $x$ such that $f(x) + v\cdot x= 0$ and $f(x)$ is different from 0
\end{enumerate}
then $\mathcal{C}_f$ defined as in \eqref{Def:Code} is a
$[q^n-1,n+1]$ minimal code over $\F_q$.
\end{theorem}

    \begin{proof}

The parameters of the code are clear by Proposition
\ref{Prop:parameters} and Remark \ref{rmk:nonlinear}.

    Let $c(u,v)$ and $c(u^{\prime},v^{\prime})$ be two nonzero codewords, with $c(u,v)\neq \lambda c(u^{\prime},v^{\prime})$ for any $\lambda \in \mathbb{F}_q^*$. Suppose that ${{\rm supp}}(c(u^{\prime},v^{\prime}))\subset {{\rm supp}}(c(u,v))$, \ that \ is \  $\overline{{\rm supp}}(c(u,v))\subset \overline{{\rm supp}}(c(u^{\prime},v^{\prime}))$.
    \begin{itemize}
        \item Suppose $v=0$. Then $\overline{{\rm supp}}(c(u,v))$ consists of all nonzero $x$ such that $f(x)=0$, i.e. $\overline{{\rm supp}}(c(u,v))=V(f)^*$. Since $\overline{{\rm supp}}(c(u,v))\subset \overline{{\rm supp}}(c(u^{\prime},v^{\prime}))$,
        $$u'f(x)+v'\cdot x=0$$
        for all $x\in V(f)^*$, that is $v'\cdot x=0$ for all $x\in V(f)^*$. Therefore, if $v'\ne 0$, then $V(f)^*\subset H(v')$, which gives a contradiction to the fact that $V(f)^*$ is $n$-dimensional. Hence $v'=0$ (and $u^\prime=\lambda u$).
        \item Suppose $v^{\prime}=0$. Then $\overline{{\rm supp}}(c(u',v'))$ consists of all nonzero $x$ such that $f(x)=0$, i.e. $\overline{{\rm supp}}(c(u',v'))=V(f)^*$.
        If $v\neq 0$, then $\overline{{\rm supp}}(c(u,v))\subset \overline{{\rm supp}}(c(u^{\prime},v^{\prime}))$ and {\rm (b)} implies that $H(v)^*\subset V(f)^*$, which is in contradiction with {\rm (a)}. So $v=0$ and $u=\lambda u$.

        \item Suppose $v,v^{\prime}\neq 0$. In this case $\overline{{\rm supp}}(c(u,v))\subset \overline{{\rm supp}}(c(u^{\prime},v^{\prime}))$ reads as
        \[
        u f(x)+v\cdot x=0\,\Rightarrow\,  u'f(x)+v'\cdot x=0
        \]
        for each $x\in\F_{q}^n\setminus\{0\}$. If $H(v)\ne H(v')$, we have a contradiction by Lemma \ref{lemma:blocking}. So $v'= \lambda v$. Taking $x$ such that $f(x)\neq 0$ (it exists by Remark \ref{rmk:nonlinear}), we get easily that $u'=\lambda u$.
    \end{itemize}
    Then ${{\rm supp}}(c(u^{\prime},v^{\prime}))\not \subset {{\rm supp}}(c(u,v))$ and $\mathcal{C}_f$ is minimal.
    \end{proof}

    \begin{remark}
    \label{rem:CONDITION}
    A necessary condition for {\rm (b)} to be true is the following:
    \begin{enumerate}
        \item[{\rm (c)}] for every non-zero vector $v$, $H(v) \cup V(f)$ is different from the whole space.
    \end{enumerate}
If $q=2$, clearly condition {\rm (c)} implies {\rm (b)}. Otherwise,
one possible easy-to-handle condition to have {\rm (b)} (assuming
{\rm (c)}) consists in asking $f$ to be constant on the lines with
homogeneous equation. In this case, since there exists $x$ such that
$v\cdot x$ and $f(x)$ are different from zero, then there exists a
non-zero scalar $\lambda$ such that $f(\lambda x)+\lambda v\cdot x=
f(x)+\lambda v \cdot x=0$ (it is enough to take $\lambda=-(v\cdot
x)^{-1}f(x))$). For example, one can choose an homogeneous
polynomial $P$, consider the set of its zeros $V(P)$ and then take
$f$ to be the characteristic function of $V(P)$.
    \end{remark}


   \begin{remark}
    For $q$ odd, let $f:\F_q^n\to \F_q$ be the function introduced in {\rm \cite{BB2019}}, that is
    \[
    f(x) =
    \begin{cases}
    \alpha_i, & {\rm wt}(x)=i\le k,\\
    0,& {\rm wt}(x)> k,\\
    \end{cases}
    \]
    where $n$ and $k$ are integers such that $n>3$ and $k\in \{2,\ldots,n-2\}$, and  $\{\alpha_i\}_{i\in\{1,\ldots,k\}}$ are $($not necessarily distinct$)$ elements of $\mathbb{F}_q^*$.

    The set $V(f)^*$ is an $n$-dimensional cutting vectorial $(1,n-1)$-blocking set in $\mathbb{A}(\F_q^n)$ :

\begin{itemize}
    \item Firstly, $V(f)^*\cap H(v)$ cannot be contained in any other $H(v^\prime)$ $($without loss of generality suppose that $v_n\ne0)$: since $v_n\neq 0$, then $x\in H(v)$ if and only if $x_n=-v_n^{-1}\cdot \sum_{i=1}^{n-1} x_iv_i$. If $x_i\neq 0$ for all $i\in\{1,\ldots,n-1\}$ and $x_n=-v_n^{-1}\cdot \sum_{i=1}^{n-1} x_iv_i$, then $x\in H(v)\cap V(f)^*$. If $x\in H(v')$, then

    $$0=\sum_{i=1}^n x_iv_i^\prime=v'_nx_n+\sum_{i=1}^{n-1} x_iv_i^\prime=-v'_nv_n^{-1}\cdot \sum_{i=1}^{n-1} x_iv_i+\sum_{i=1}^{n-1} x_iv_i^\prime$$
    that is
     $$\sum_{i=1}^{n-1} (v'_nv_i-v_nv_i^\prime)x_i=0.$$
    Since the last equality should hold for any $x$ such that $x_i\neq 0$ for all $i\in\{1,\ldots,n-1\}$, in particular it holds for $x$ with $x_1=\ldots=x_{n-1}=1$ and for  $\overline{x}^{(j)}$ with $\overline{x}^{(j)}_1=\ldots=\overline{x}^{(j)}_{j-1}=\overline{x}^{(j)}_{j+1}=\ldots=\overline{x}^{(j)}_{n-1}=1$ and $\overline{x}^{(j)}_j=-1$, for any $j\in \{1,\ldots,n-1\}$. We get, for every $j\in \{1,\ldots,n-1\}$,
    $$2(v'_nv_j-v_nv_j^\prime)=0$$
     and, since $q$ is odd,
    $v_j^\prime=v_n^{-1}v'_nv_j$
    for every $j\in \{1,\ldots,n\}$, so that $H(v)=H(v')$.
\item Secondly, we show that $V(f)^*$ does not contain any hyperplane without the origin $H(v)^*$ (without loss of generality suppose again that $v_n\ne 0$): $H(v)$ always contains a vector $u$ such that ${\rm wt}(u)\le2$ (in this case $f(u)\neq 0$, so $u\not\in V(f)^*$), in particular if ${\rm wt}(v)=1$ then $e_1\in H(v)$ and if there exists $j\in\{1,\dots,n-1\}$ such that $v_j\ne0$, then $u=v_ne_j-v_je_n\in H(v).$ In both cases we found a vector in $H(v)\setminus V(f)^*$.
\end{itemize}

In this case, condition {\rm (b)} of Theorem \ref{Th:Main} is
clearly verified.

    \end{remark}

    \subsection{Projective case}

    We may extend the general construction of the code done above to the projective space, since it can be easily adapted to the projective case without notable differences.

    In this case, in the definition of the code, instead of considering all vectors in $\F_q^n\setminus \{0\}$, we will take a nonzero vector from each $1$-dimensional subspace of $\F_q^n$ (that is one representative of each projective point in $\mathbb{P}(\F_q^n)$). This is a quite standard choice in coding theory (the same that one usually does to define simplex and Hamming codes in the non-binary case \cite{HP}).

    For every homogeneous (polynomial) function $f:\mathbb{F}_q^{n}\to \mathbb{F}_q$, let $\widetilde{\mathcal{C}}_F$ be the linear code defined as
    \begin{equation}\label{Def:Code_Proj}
    \widetilde{\mathcal{C}}_f :=\{(uf(x)+v\cdot x)_{x \in \mathbb{P}(\F_q^{n})} \ \mid u\in\F_q\,, v \in \mathbb{F}_q^{n}\}.
    \end{equation}

    Here we are supposing, as explained above, to have a fixed ordering of the points in $\mathbb{P}(\F_q^{n})$ and to chose one representative for each projective point. What follows will not depend on these choices: different choices yield equivalent codes.
    We keep the notations as in the affine case.

If $f$ is non-linear, then $\widetilde{\mathcal{C}}_f$ is a
$[(q^n-1)/(q-1),n+1]$ code (the proof is exactly the same as for
Proposition \ref{Prop:parameters}). Moreover, analogously to the
affine case, a generator matrix of $\widetilde{\mathcal{C}}_f$ can
be obtained by extending the generator matrix of a simplex code by a
line from the evaluation of the homogeneous function $f$. Therefore,
${\rm wt}(c(0,v))=(q^n-q^{n-1})/(q-1)$ and ${\rm
wt}(c(u,0))=(q^n-1)/(q-1)-\#V_p(f)$.

An analogue of Lemma \ref{lemma:AB} holds.

\begin{lemma}\label{lemma:AB_Proj}
If $\displaystyle \# V_p(f)\geq \frac{2q^{n-1}-q^{n-2}-1}{q-1}$,
then
\[
\frac{w_{max}}{w_{min}} \geq \frac{q}{q-1}.
\]
\end{lemma}

Clearly, for a homogeneous function $f$, $V(f)^*$ satisfies the
condition of Lemma \ref{lemma:AB} if and only if $V_p(f)$ satisfies
the condition of Lemma \ref{lemma:AB_Proj}.

\begin{remark}
The projective context seems to be a more natural scenario: there
are no restricitions on hyperplanes and, as we mentioned already, in
the projective case the simplex code is a subcode of codimension $1$
in $\widetilde{\mathcal{C}}_f$. On the other hand, we have less
freedom in the choice of the function.
\end{remark}

Finally, we state, without proving it, the projective analogue (note
that the dimensions drop by $1$) of Theorem \ref{Th:Main}. The proof
works exactly in the same way as in the affine case.

\begin{theorem}\label{Th:Main_Proj}:
Let $f:\F_q^n\to \F_q$ be a homogeneous function. If

\begin{enumerate}
    \item[{\rm(a)}] $V_p(f)$ is an
$(n-1)$-dimensional cutting projective $(1,n-2)$-blocking set in
$\mathbb{P}(\F_q^n)$;
\item[{\rm (b)}] for every non-zero vector $v$, it exists  $x$ such that $f(x) + v\cdot x= 0$ and $f(x)$ is different from 0,
\end{enumerate}
\noindent then $\widetilde{\mathcal{C}}_f$ defined as in
\eqref{Def:Code_Proj} is a $[(q^n-1)/(q-1),n+1]$ minimal code over
$\F_q$.
\end{theorem}

\bigskip
    \section{An example: a family of good functions}\label{sec:functions}

    Let $n=rk$ be a positive integer and consider the hypersurface $V(f_{r,k})$ of $\mathbb{A}(\F_{q}^{n})$ defined by the (polynomial) function $f_{r,k}:\F_q^n\to \F_q$

    $$f_{r,k}(x_1,\ldots,x_n):=\sum_{j=0}^{k-1} x_{jr+1}x_{jr+2}\cdots x_{jr+r}.
    $$

    \begin{theorem}\label{Prop:properties}
    If $k\geq 2$ and $r\geq 2$, then $V(f_{r,k})^*$
    is an $n$-dimensional cutting vectorial $(1,n-1)$-blocking set  and for every non-zero vector $v$, it exists  $x$ such that $f(x) + v\cdot x= 0$ and $f(x)$ is different from 0.
    \end{theorem}

    \begin{proof}
 Let $H(v)$ be a hyperplane. It exists $i$ such that $v_i\neq 0$. To simplify the notation, suppose, without loss of generality, that $i=n$, so that $x_n=-(v_n)^{-1}(v_1x_1+\ldots+v_{n-1}x_{n-1})$.

 For all $i\in \{1,\ldots,n-1\}$, let
 $$p_i:=e_i-(v_n)^{-1}v_ie_n=(0,\ldots,0,1,0,\ldots,0,-(v_n)^{-1}v_i)$$
and let
$$\overline{p}:=p_1+\ldots+p_r=(\underbrace{1,\ldots,1}_{r \text{ times}},0,\ldots,0,-(v_n)^{-1}(v_1+\ldots+v_{r})).$$
\medskip

We observe that $\overline{p}\in H(v)^*$ but not in $V(f_{r,k})^*$.
So $H(v)^*\not\subset V(f_{r,k})^*$ (property {\rm (b)}).
\medskip

Let us prove that property {\rm (a)} holds too.

Let $r>2$. In this case, $\{p_1,\ldots,p_{n-1}\}\subset H(v)\cap
V(f_{r,k})$. Suppose that $H(v)\cap V(f_{r,k})\subset H(v')$. Then
$v'\cdot p_i=v_i'-((v_n)^{-1}v'_n)v_i=0$ for all
$i\in\{1,\ldots,n-1\}$. Clearly $v_n'-((v_n)^{-1}v'_n)v_n=0$. So
$v'=((v_n)^{-1}v'_n)v,$ which implies $H(v')=H(v)$.
\medskip

Let $r=2$. In this case, $\{p_1,\ldots,p_{n-2}\}\subset H(v)\cap
V(f_{2,k})$. Suppose that $H(v)\cap V(f_{2,k})\subset H(v')$. Then,
as above, $v'_i=((v_n)^{-1}v'_n)v_i$ for $i\in \{1,\ldots,n-2,n\}$.
If $H(v)\ne H(v')$, this means that $v'=v+ke_{n-1}$, with $k\neq 0$.
There are three cases to be considered: if $v_{n-1}=0$, then
$e_{n-1}\in H(v)\cap V(f_{2,k})$ but not in $H(v')$; if $v_{n-1}\ne
0$ and it exists $j\in \{1,\ldots,n-2\}$ such that $v_j\ne 0$, then
$v_{n-1}e_j-v_je_{n-1}$ in $H(v)\cap V(f_{2,k})$ but not in $H(v')$;
if $v_{n-1} \ne 0$ and $v_j=0$ for all $j\in\{1,\ldots,n-2\}$, then
$v_ne_1+v_{n-1}e_2-v_ne_{n-1}+v_{n-1}e_n \in H(v)\cap V(f_{2,k})$
but not in $H(v')$. In all three cases we get then a contradiction.

In order to prove that for every non-zero vector $v$, it exists  $x$
such that $f(x) + v\cdot x= 0$ and $f(x)$ is different from 0, let
$\{i_1,..,i_r\}$ be the support of the vector $v$. We will have that
$v\cdot x=v_{i_1}x_{i_1}+...+ v_{i_r}x_{i_r}$. Without loss of
generality we can consider $i_1=1$, $v_1=1$ and $x$ to be such that
the following relation holds (since $x$ does not annihilate $v\cdot
x$)
\[
x_{1}+ v_{i_2}x_{i_2} +...+ v_{i_r}x_{i_r} =1
\]
and so
\[
x_{1} =1-(v_{i_2}x_{i_2} +...+ v_{i_r}x_{i_r}).
\]
If we take $\hat{x}$ having its components as follows
\[
\hat{x}_i=\begin{cases}
1-(v_{i_2}x_{i_2} +...+ v_{i_r}x_{i_r})\quad \text{ for } i=1\\
0\qquad \qquad \qquad \qquad \qquad \quad \text{ for } i\in\{2,\dots,r(k-1)\}\\
1\qquad \qquad \qquad \qquad \qquad \quad \text{ for }
i\in\{r(k-1)+1,\dots,rk-1\}
\\-1\quad \qquad \qquad \qquad \qquad\quad\,  \text{ for } i=rk\\\end{cases}
\]
then $\hat{x}$ is such that $f(\hat{x})=-1\ne0$, $v\cdot
\hat{x}=1\ne0$ and $f(\hat{x})+v\cdot \hat{x}=0$.

    \end{proof}






    \begin{corollary}
    If $r\geq 2$ and $k\geq 2$, the code $\mathcal{C}_{f_{r,k}}$ is minimal.
    \end{corollary}

    \begin{proof}
    It follows directly by Theorem \ref{Prop:properties} and Theorem \ref{Th:Main}.
    \end{proof}

    \begin{lemma}\label{Prop:number}
The cardinality of $V(f_{r,k})$ is
    \begin{equation}\label{eq:number}
        \#V(f_{r,k})=(q-1)\cdot q^{k-1}\cdot (q^{r-1}-(q-1)^{r-1})^k+q^{rk-1}.
    \end{equation}
    \end{lemma}

    \begin{proof}
    Let us start considering the monomial $x_1x_2\dots x_r$. It is clearly non-zero for $(q-1)^r$ choices of $(x_1,x_2,\dots,x_r)$, and then it is zero for $q^r-(q-1)^r$ choices. Moreover, it takes all non-zero values equally often, so for a $t\in\F_q$ we have that it assumes the value $t$ exactly $(q-1)^{r-1}$ times.
    Let us fix $r$ and define $Z(k):=\#V(f_{r,k})$.
    Note that
    $$V(f_{r,k})=V(f_{r,k-1})\times V(f_{r,1})+$$
    $$+(\F_q^{(k-1)r}\setminus V(f_{r,k-1}))\times\{(x_1,\ldots,x_r)\ | \ f_{r,1}(x_1,\ldots,x_r)=t\neq 0\}$$
 So we get the following recursive formula to compute $Z(k)$:
    \[
    \begin{split}
        Z(k)=&Z(k-1)\cdot(q^r-(q-1)^r)+(q^{(k-1)r}-Z(k-1))\cdot (q-1)^{r-1}\\
        =&Z(k-1)\cdot q\cdot (q^{r-1}-(q-1)^{r-1})+q^{(k-1)r}\cdot (q-1)^{r-1}\\
    \end{split}
    \]

    Let us prove \eqref{eq:number} by induction on $k$.

    For $k=1$, we have
    $$Z(1)=(q-1)\cdot (q^{r-1}-(q-1)^{r-1})+q^{r-1}=q^r-(q-1)^r,$$
 which is the correct number, as observe above. So we focus on the inductive step:

    \[
    \begin{array}{rl}
        Z(k)=&Z(k-1)\cdot q\cdot (q^{r-1}-(q-1)^{r-1})+q^{(k-1)r}\cdot (q-1)^{r-1}\\
=&((q-1)\cdot q^{k-2}\cdot(q^{r-1}-(q-1)^{r-1})^{k-1}+\\
        &+q^{r(k-1)-1})\cdot q\cdot (q^{r-1}-(q-1)^{r-1})+q^{(k-1)r}\cdot(q-1)^{r-1}\\
        =&(q-1)\cdot q^{k-1}\cdot(q^{r-1}-(q-1)^{r-1})^{k}+\\
        &+q^{r(k-1)}\cdot(q^{r-1}-(q-1)^{r-1})+q^{r(k-1)}\cdot(q-1)^{r-1}\\
        =&(q-1)\cdot q^{k-1}\cdot(q^{r-1}-(q-1)^{r-1})^{k}+q^{rk-1}
    \end{array}
    \]
    \end{proof}

\begin{lemma}\label{Prop:AB}
If $r\geq
2+\log_{\left(1-\frac{1}{q}\right)}\left(\frac{q-\sqrt{q}}{q-1}\right)$,
then
$$\#V(f_{r,2})^*\geq 2q^{2r-1}-q^{2r-2}-1.$$
\end{lemma}

\begin{proof}
By Lemma \ref{Prop:number},
        $$\#V(f_{r,2})^*=(q-1)\cdot q\cdot (q^{r-1}-(q-1)^{r-1})^2+q^{2r-1}-1.$$
    We want to determine when
    $$\#V(f_{r,2})^*\geq 2q^{2r-1}-q^{2r-2}-1.$$

    This is equivalent to solve
    $$q^{2r-2}(q-1)^2+q(q-1)^{2r-1}-2q^r(q-1)^r\ge 0$$
 and with some manipulations we obtain
    $$q^{2r}\cdot\left(1-\frac{1}{q}\right)^2\cdot\left(\left(1-\frac{1}{q}\right)\cdot\left(\left(1-\frac{1}{q}\right)^{r-2}\right)^2-2\cdot\left(1-\frac{1}{q}\right)^{r-2}+1\right)\ge 0$$
    The first two factors are positive. To study the third one, we can call $t=\left(1-\frac{1}{q}\right)^{r-2}$, and solve
    $$\left(1-\frac{1}{q}\right)\cdot t^2-2\cdot t+1\geq 0$$
    which gives
    $$t\leq \frac{q-\sqrt{q}}{q-1} \ \text{ or } \ t\geq \frac{q+\sqrt{q}}{q-1}$$
    that is,
    $$r\geq 2+\log_{\left(1-\frac{1}{q}\right)}\left(\frac{q-\sqrt{q}}{q-1}\right)$$
\end{proof}

\begin{theorem}
    There are infinitely many values of $r$ such that $\mathcal{C}_{f_{r,2}}$ is minimal and does not satisfy the AB condition.
\end{theorem}

\begin{proof}
    If $r\geq 2+\log_{\left(1-\frac{1}{q}\right)}\left(\frac{q-\sqrt{q}}{q-1}\right)$, then $C_{f_{r,2}}$ is minimal, by Lemma \ref{Prop:properties} and Theorem \ref{Th:Main}, and does not satisfy the AB condition, by Lemma \ref{Prop:AB} and Lemma \ref{lemma:AB}.
\end{proof}

\begin{remark}
Note that the function defined above is homogeneous, so that we can
consider it in the projective case. Completely analogously to the
affine case, we have that $\widetilde{\mathcal{C}}_{f_{r,k}}$ is
minimal for $r\geq 2$ and $k\geq 2$, and there are infinitely many
values of $r$ such that $\widetilde{\mathcal{C}}_{f_{r,2}}$ does not
satisfy the AB condition.
\end{remark}

\bigskip
\section{Conclusion and open problems}\label{sec:conclusion}
In this paper we have linked a well-known construction of minimal
linear codes to the geometry of its defining function, which was a
quite unexplored topic and which may lead to new and interesting
constructions in coding theory. A new notion in blocking sets'
theory is introduced and related examples are provided. Moreover, we
present an infinite family of minimal linear codes not satisfying
the Ashikhmin-Barg's condition.

To conclude, we list here some of the possible developments of our
results.

\begin{problem}
Give a necessary and sufficient conditions for $C_f$ to be minimal,
which relies only on the geometry of $V(f)$.
\end{problem}

\begin{problem}
Investigate the weight distribution of $C_f$ for some cutting
blocking sets.
\end{problem}

\begin{problem}
In \S\ref{sec:functions} we introduced the family of functions
$f_{r,k}$ and we proved that infinitely many codes are minimal when
$k=2$. From computations, it seems that the same holds for $k>2$ but
the inequalities become much harder to be solved.
\end{problem}

We think that our geometrical approach gives a deeper insight into
the problem and may lead to other future developments and
constructions.
\bigskip
\section*{Acknowledgments}
The research of M. Bonini was supported by the Italian National
Group for Algebraic and Geometric Structures and their Applications
(GNSAGA - INdAM). The authors gratefully thank M. Giulietti for the
fruitful discussions about the paper.

\bigskip


\end{document}